\documentclass{new_tlp}
\usepackage{breakurl}             
\usepackage{underscore}           
\usepackage[T1]{fontenc}
\usepackage{lmodern}
\usepackage{amsmath}
\usepackage{amssymb}
\usepackage{amsfonts}

\usepackage{bussproofs}
\usepackage{multicol}
\usepackage{listings}
\usepackage[textsize=tiny]{todonotes}
\usepackage{url}

\newtheorem{lemma}{Lemma}
\newtheorem{definition}{Definition}
\newtheorem{theorem}{Theorem}
\newtheorem{proposition}{Proposition}
\newtheorem{example}{Example}

\def\type{type}
\def\term{term}
\def\termSPLIT{te&rm}
\DeclareMathOperator{\app}{app}

\newcommand{\transl}[1]{\ulcorner #1 \urcorner}

\DeclareMathOperator{\mvar}{mvar}
\DeclareMathOperator{\mtvar}{mtvar}

\DeclareMathOperator{\sig}{\mathcal{S}}

\DeclareMathOperator{\bool}{\mathtt{bool}}
\DeclareMathOperator{\ttb}{\mathtt{tt}}
\DeclareMathOperator{\ffb}{\mathtt{ff}}
\DeclareMathOperator{\eqb}{{\equiv_{\bool}}}
\DeclareMathOperator{\elimeqb}{{\mathtt{elim}_{\eqb}}}
\DeclareMathOperator{\A}{\mathtt{A}}
\DeclareMathOperator{\maybeA}{\mathtt{maybe_A}}
\DeclareMathOperator{\elimmaybeA}{\mathtt{elim}_{\maybeA}}

\DeclareMathOperator{\elim}{{\mathtt{elim}}}
\DeclareMathOperator{\typeK}{{\mathtt{type}}}
\DeclareMathOperator{\kindS}{{\mathtt{kind}}}

\usepackage{layouts}
 
\usepackage{myabbrev}

\newcommand{\shift}[1]{#1 \! \uparrow }
\newcommand{\shiftn}[2]{#1 \! \uparrow ^{ #2 } }
\newcommand{\subst}[2]{#1 [ #2 ] }
\newcommand{\substn}[3]{#1 [ #2 / #3 ] }
\newcommand{\shifteq}[2]{#1 \! \uparrow \; \equiv \; #2}
\newcommand{\shifteqI}[3]{#1 \! \uparrow^{#2} \; \equiv \; #3}
\newcommand{\substeq}[3]{#1 [ #2 ] \equiv #3 }
\newcommand{\substeqI}[4]{#1 [ #2 / #3 ] \equiv #4 }
\def\whr{\overset{\text{whr}}\longrightarrow}

\title[
    Proof-relevant Horn Clauses for Dependent Type Inference and Term Synthesis]
{
    Proof-relevant Horn Clauses for Dependent Type Inference and Term Synthesis}

\author[F. Farka, E. Komendantskaya, and K. Hammond]{
    Franti\v{s}ek Farka\\
    University of St Andrews, and Heriot-Watt University\\
    \email{ff12@st-andrews.ac.uk}
    \and Ekaterina Komendantskya\\
    Heriot-Watt University\\
    \email{ek19@hw.ac.uk}
    \and Kevin Hammond\\
    University of St Andrews\\
    \email{kevin@kevinhammond.net}
}

\newcommand{\ident}[1]{{\color{blue}\texttt{#1}}}
\newcommand{\identix}[2]{{\color{blue}\texttt{#1}_{\texttt{#2}}}}
\newcommand{\kv}[1]{\textbf{#1}}
\newcommand{\constr}[1]{{\color{green!50!black}\texttt{#1}}}

\begin{document}



\maketitle

\begin{abstract}
First-order resolution has been used for type inference for many years, including
in Hindley-Milner type inference, type-classes, and constrained data types.
Dependent types are a new trend in functional languages.
In this paper, we show that proof-relevant first-order resolution can play an important role in automating
type inference and term synthesis for dependently typed languages.
We propose a calculus that translates
 type inference and term synthesis problems in a dependently typed language
to a logic program and a goal in the proof-relevant first-order Horn clause logic.
The computed answer substitution and proof term then provide a solution to the given type
inference and term synthesis problem.
We prove the decidability and soundness of our method.
The paper is under consideration for acceptance in TPLP.
\end{abstract}

%


\section{Introduction}
\label{sect:intro}

First-order resolution is well known for supporting a range of automated reasoning
methods for type inference. Simple types have been a part of mainstream
languages since the 1960s; polymorphic types have been available in advanced
languages such as ML and Haskell since the 1980s; and type classes were
introduced from the 1990s onwards. Logic programming has had a role to play in each of these
stages.
Hindley and \citeN{Milner78} were the first to notice that type inference in
simply typed lambda calculus can be expressed as a first-order unification problem.
For example, the rule for term application in this calculus
\begin{prooftree}
    \AxiomC{$\Gamma \vdash M : A \rightarrow B$}
    \AxiomC{$\Gamma \vdash N : A$}
    \RightLabel{\textsc{App}}
    \BinaryInfC{$ 
        \Gamma \vdash M N : B
    $}
    \label{ir:app}
\end{prooftree}
gives rise to a type inference problem encoded by the following Horn clause:
\begin{equation*}
\type(\Gamma, \app(M, N), B) \leftarrow 
    \type(\Gamma, M, A \rightarrow B) \land \type(\Gamma,N,A)
\end{equation*}

\noindent
Given a term $E$, the query $\type(\Gamma,E,T)$  infers a type $T$ in a context $\Gamma$ such
that the typing judgement  $\Gamma \vdash E : T$ holds.
%
This general scheme allows a multitude of extensions. For example, in $HM(X)$
type inference \cite{OderskySW99}, constrained types must be accounted for.
For this
extension, a constraint logic programming $CLP(X)$ \cite{SulzmannS08}, was
suggested, in which constraint solving over a domain $X$ was added to the
existing first-order unification and resolution algorithms.
%
%
Haskell
type classes are
another example of the application of Horn clause resolution.  
Consider the following instances of the Haskell equality type
class:
%
\begin{lstlisting}
    instance Eq Int where ...
    instance(Eq x,Eq y) => Eq(x,y) where ...
\end{lstlisting}
%
They can be encoded by the following Horn clauses, annotated with 
names $\kappa_\mathtt{Int}$ and $\kappa_\mathtt{Pair}$:
\begin{equation*}
    \begin{aligned}
        \kappa_{\mathtt{Int}} && :~& eq(int)\\
        \kappa_{\mathtt{Pair}}   && :~&  eq(X,Y)  \leftarrow eq(X) \land eq(Y)
    \end{aligned}
\end{equation*}
Type class instance resolution is then implemented as first-order resolution
on Horn clauses. There is only one caveat---a \emph{dictionary} (that is, a proof term)
needs to be constructed \cite{PeytonJonesJM97}. 
For example, $Eq~ (Int, Int)$ is inferred to have a dictionary
$\kappa_\mathtt{Pair}(\kappa_\mathtt{Int}, \kappa_\mathtt{Int})$.
This 
records the resolution trace and
is treated as a witness of the type class  instance $Eq~(Int, Int)$.
Horn clause resolution is thus extended to \emph{proof-relevant resolution} \cite{FuKSP16}. This
line of work is on-going: various extensions to the syntax of type
classes are still being investigated \cite{KarachaliasS17}.
In recent years, the idea of relational type inference has been taken further by
miniKanren \cite{HemannFBM16}. This offers a range of relational domain specific languages
for ML, Rust, Haskell and many other languages. As
\citeN{AhnV16} point out, a relational language can be very convenient in
encodings of type inference problems.

In the last decade, 
dependent types \cite{WeirichVAE17,Brady13} have gained popularity in the programming language community.
They allow reasoning about program values
within the types, and thus give more general, powerful and flexible
mechanisms to enable verification of both the functional (correctness,
compliance, etc.) and the non-functional (execution time,
space, energy usage etc.) properties of code. 
Automation of type inference represents a big challenge for these
languages. Most dependently typed languages, such as Coq, Agda or Idris,
incorporate a range of algorithms that automate various aspects of type
inference (\cf \citeN{Pientka13}).
Some use reflection \cite{SlamaB17},
some are based on algorithms that are
similar to first-order resolution \cite{GonthierM10},
and others (\eg Liquid Haskell) incorporate
third-party SMT solvers \cite{VazouTCSNWJ18}. However, to the best of our knowledge, logic
programming has not yet made its definitive mark in this domain.

This paper fills this gap: we propose a first systematic
approach to logic-programming based type inference for a dependently-typed
language.
We demonstrate that Horn clause logic provides a convenient formal language
to express type inference problems while staying very close to the
formal specification of the dependently-typed language. Proof-relevant resolution
then
computes proof terms that capture well-formedness derivations 
of objects in the language.
We present a method to synthesise terms of the dependently typed language from such proof terms.
This  method can be applied in a more complex setting where a small kernel of a verified 
compiler off-loads proof-relevant resolution to an external, non-verified 
resolution engine
and then verifies synthesised derivations internally.

In the next section, we explain our main idea by means of an example.

\subsection*{Overview  of Results by Means of  an Example}
We rely on an abstract syntax that closely resembles existing functional
programming languages with dependent types. We will call it the \emph{surface language}.  
%
%
    Using the syntax we define $\identix{maybe}{A}$, an option type over a fixed type
    $\ident{A}$, indexed by a Boolean:
     \begin{equation*}
        \begin{aligned}
            \kv{data}~&\identix{maybe}{A}~(a : \ident{A}) && : \ident{bool} \to \typeK ~\kv{where}\\
            & \constr{nothing} && : \identix{maybe}{A}~\constr{ff}\\
            & \constr{just} && : \ident{A} \to \identix{maybe}{A}~\constr{tt} 
        \end{aligned}
    \end{equation*}
\noindent
Here, $\constr{nothing}$ and  $\constr{just}$ are the two \emph{constructors} of the $\ident{maybe}$ type.
The type is indexed by $\ident{ff}$ when the $\constr{nothing}$ constructor is used, and by $\ident{tt}$ when the $\constr{just}$ constructor is used
($\constr{ff}$ and $\constr{tt}$ are constructors of $\ident{bool}$).
A function $\ident{fromJust}$ extracts the
value from the $\constr{just}$ constructor:
    \begin{equation*}
        \begin{aligned}
            \ident{fromJust} & ~:~ \identix{maybe}{A}~\constr{tt} \to \ident{A}\\
            \ident{fromJust} & ~(\constr{just}~x) = x
        \end{aligned}
    \end{equation*}

    \noindent
    Note that the value $\constr{tt}$ appears within the type $\identix{maybe}{A}~\constr{tt} \to \ident{A}$ of this function (the type \emph{depends} on the value), allowing for
    a more precise function definition that omits the redundant case when the
    constructor of type $\identix{maybe}{A}$ is $\constr{nothing}$.
The challenge for the type checker is to determine that the missing case
$\ident{fromJust} ~\constr{nothing}$ 
is contradictory (rather than being omitted by mistake). Indeed, the type of $\constr{nothing}$
 is $\identix{maybe}{A}~\constr{ff}$. However, the function specifies its argument to be of type $\identix{maybe}{A}~\constr{tt}$.

To type check such functions, the compiler translates them
into  terms in a type-theoretic calculus. 
In this paper, we will rely on the calculus LF~\cite{HarperP05}, a standard and well-understood first-order dependent type theory.
We call this calculus the \emph{internal language} of the
compiler.
For our example, the signature of the  internal language is as follows:
%
\begin{multicols}{2}
    \begin{lstlisting}%[basicstyle=\footnotesize]

 A      : type
 bool   : type
 ff tt  : bool

 (==) : bool -> bool -> type
 refl   : Pi(vb:bool).vb == vb
 elim_==: tt == ff -> A

 maybe_A  : bool -> type
 nothing : maybe_A ff
 just    : A -> maybe_A tt
 elim_maybe_A: Pi(vb:bool).maybe_A vb
           -> (vb == ff -> A)
           -> (vb == tt -> A -> A) 
           -> A
\end{lstlisting}
\end{multicols}

\noindent 
We use $A \to B$ as an abbreviation for $\Pi (a : A) . B$ where $a$ does not
occur free in $B$. The final goal of type checking of the function
$\ident{fromJust}$ in the 
surface language 
is to obtain the following encoding 
in the internal language:

\begin{lstlisting}
    t_fromJust := \ ( vm:maybe_A tt).elim_maybe_A tt vm
                    (\ (vw:tt==ff).elim_== vw)
                    (\ (vw:tt==tt).\ (vx:A).vx)
\end{lstlisting}


%
%
\noindent
The missing case for 
$\constr{nothing}$ must be accounted for
(\cf the line  \lstinline{ (\ (vw:tt==ff).elim_== vw)} above).
In this example (as is generally the case), only partial information is given in the surface language.
To address this problem, we extend the internal language with term level metavariables, denoted by $?_a$, and
type level metavariables, denoted by $?_A$.  These stand for the parts of a term in the
internal language that are not yet known. Using metavariables, the term that directly
corresponds to $\ident{fromJust}$ is: 

\vskip-1em
\begin{minipage}[t]{4.7in}
\begin{lstlisting}
    t_fromJust := \ ( vm:maybe_A tt).elim_maybe_A Ma vm
                    (\ ( vw:  MA  ).  Mb  )
                    (\ ( vw:  MB  ).\ ( vx:A).vx)
\end{lstlisting}
\end{minipage}
\begin{minipage}[t]{1cm}
\begin{equation}
    ~   
    \label{ex:tjust}
    \tag{I}
\end{equation}
\end{minipage}

\noindent
The missing information comprises the two types  $?_A$ and $?_B$ and the term $?_b$ for the constructor
$\constr{nothing}$. 
Obtaining types $?_A$, $?_B$
amounts to type inference
(in the internal language, as opposed to checking in the surface language),
whereas obtaining the term $?_b$ amounts to
term synthesis.
In this paper, we are interested in automating such reasoning.
We use the notion \emph{refinement} to refer to the 
combined problem of type inference and term synthesis.
We make use of \emph{proof-relevant} Horn clause logic~\cite{FuK17} that was initially inspired by type class resolution, as described above.  In this logic, Horn clauses are seen as
types and proof witnesses --- as terms inhabiting the types.
Given a proposition---a \emph{goal}---and a set of Horn clauses---a
\emph{logic program}---the resolution process
is captured by an explicit proof term construction. 
We translate \emph{refinement problems} 
into the syntax of  logic programs.
The \emph{refinement} algorithm that we propose takes a signature and a term with metavariables in the extended internal language
to a logic program and a goal in proof-relevant Horn clause logic.
The unifiers that are computed by resolution give an assignment of types to type-level metavariables.
At the same time, the computed proof terms are interpreted
as an assignment of terms to term-level metavariables. 
We illustrate the process in the following paragraphs.

Consider the inference rule $\Pi$\textsc{-t-Elim} in LF, which generalises the inference rule \textsc{App} given on page 
\pageref{ir:app} of this paper:
\begin{equation*}
    \AxiomC{$ \Gamma \vdash M : \Pi x : A . B  $}
    \AxiomC{$ \Gamma \vdash N : A $}
    \RightLabel{\sc $\Pi$-t-Elim}
    \BinaryInfC{$
        \Gamma \vdash M N : B [ N / x]
    $}
    \DisplayProof
\end{equation*}
%
%
%
%
%
When type checking the term $t_{\mathtt{fromJust}}$ (defined in (I))
an application of
$\elimmaybeA~\ttb~m$ to the term
$\lambda (w : ?_A) . ?_b$ in the context 
$m : \maybeA~\ttb$
needs to be type checked. This amounts to providing 
a derivation of the typing judgement that contains the following instance of the
rule $\Pi$\textsc{-t-Elim}:
\begin{equation*}
    \AxiomC{$
        m : \maybeA~\ttb 
        \vdash
        \elimmaybeA~\ttb~m
    $}
    \noLine
    \UnaryInfC{$
        \quad \quad \quad : 
        (\ttb \eqb \ffb \to \A) \to \cdots \to  \A
    $}
    \AxiomC{$
        m : \maybeA~\ttb 
        \vdash
        \lambda (w : ?_A) . ?_b : ?_A \to ?_B
        $}
    \BinaryInfC{$
        m : \maybeA~\ttb 
        \vdash
        (\elimmaybeA~\ttb~m)~(\lambda (w : ?_A) . ?_b)
        : (\ttb \eqb \ttb \to \A \to \A) \to \A
    $}
    \DisplayProof
\end{equation*} 
%
For the above inference step to be a valid instance of the inference
rule {\sc $\Pi$-t-Elim}, it is necessary that
$(\ttb \eqb \ffb) =~ ?_A$
and $\A =~ ?_B$.
%
%
This is reflected in the goal: 
\begin{equation*}
    ((\ttb \eqb \ffb) =~ ?_A)
\land (\A =~ ?_B)
\land G_{ (\elimmaybeA~\ttb~m)}
\land
G_{\lambda (w : ?_A) . ?_b)}
\label{ex:goal}
\tag{II}
\end{equation*} 
The additional goals
$G_{ (\elimmaybeA~\ttb~m)}$ and
$G_{\lambda (w : ?_A) . ?_b}$ 
are recursively generated by the algorithm for the terms
$\elimmaybeA~\ttb~m$ and
$\lambda (w : ?_A) . ?_b$, respectively.
%
Similarly, assuming
the term $\lambda (w : ?_A) . ?_b$ is of type $(\ttb
\eqb \ffb) \to \A$, type checking places restrictions on the term $?_b$:
\begin{equation*}
    \AxiomC{$
        m : \maybeA~\ttb 
        \vdash
        \ttb \eqb \ffb : \typeK
    $}
    \AxiomC{$
        m : \maybeA~\ttb , w : \ttb
        \eqb \ffb
        \vdash
        ?_b : \A
        $}
    \BinaryInfC{$
        m : \maybeA~\ttb 
        \vdash
        \lambda (w : \ttb \eqb \ffb) .
            ?_b : \ttb \eqb \ffb \to \A
    $}
    \DisplayProof
\end{equation*}

\noindent
%
That is, $?_b$ needs to be a well-typed term of type $\A$ in a context
consisting of $m$ and $w$.
Recall that in the signature there is a constant $\mathtt{elim}_{\equiv_{\bool}}$ of type
$\ttb \equiv_{\bool} \ffb \to \A$. Our translation will turn this constant into 
a clause in the generated logic program. There will be a clause that
corresponds to the inference rule for elimination of a $\Pi$ type as well:
\begin{equation*}
    \begin{split}
        \kappa_{\elimeqb} & :~ 
        \term(\elimeqb, \Pi x : \ttb \eqb \ffb . \A , ?_\Gamma) \leftarrow \\
        \kappa_{\elim} & :~ \term(?_M ?_N, ?_B, ?_\Gamma) \leftarrow 
        \term(?_M, \Pi x : ?_A . ?_{B'}, ?_\Gamma) 
        \land
        \term(?_N, ?_A, ?_\Gamma)
        \land
        \substeq{?_{B'}}{?_N / x}{?_B}
    \end{split}
\end{equation*}
\noindent
The above clauses are written  in the  proof-relevant Horn clause logic,
and thus $\kappa_{\elimeqb}$ and $\kappa_{\elim}$ now play the role of 
proof-term symbols (``witnesses'' for the clauses).
%
In this clause, $?_M$, $?_N$, $?_A$, $?_B$, $?_{B'}$ and $?_{\Gamma}$ are \emph{logic variables}, i.e. variables of the first-order logic.
By an abuse of notation, we  use the same symbols for metavariables of the internal language and
logic variables in the logic programs generated by the refinement algorithm.
We also use the same notation for objects of the internal language and terms of
the logic programs. This is possible since
we represent variables using \emph{de Bruijn indices}.

%
The presence of $w : \ttb \eqb \ffb$
in the context allows us to use the clause $\elimeqb$ to resolve the goal
$\term(?_M ?_N, \A, [m : \maybeA~\ttb , w : \ttb
\eqb \ffb])$:
%
\begin{equation}
    \begin{split}
        \termSPLIT (?_M ?_N ,  \A , [m : {\maybeA}, w : \ttb
        \eqb \ffb])
        \rightsquigarrow_{\kappa_{\elim}} \\
        & \term(?_M, \Pi x : ?_A . \A,[\dots])
        \land
        \term(?_N, ?_A, [\dots, w : \ttb \eqb \ffb])
        \land
        \substeq{\A}{?_N / x}{?_B}
        \rightsquigarrow_{\kappa_{\elimeqb}} \\
        & \term(?_N, \ttb \eqb \ffb , [\dots, w : \ttb \eqb \ffb])
        \land
        \substeq{\A}{?_N / x}{?_B}
        \rightsquigarrow_{\kappa_{\texttt{proj}_w}} \\
        & \substeq{\A}{?_N / x}{?_B}
        \rightsquigarrow_{\kappa_{\texttt{subst}_{\A}}}
        \bot
    \end{split}
    \tag{III}
        \label{ex:deriv}
\end{equation}

\noindent
The resolution steps are denoted by $\rightsquigarrow$. Each step is indexed by the name
of the clause that was used. First, the goal is resolved in one step using the
clause $\kappa_{\elim}$. A clause $\kappa_{\texttt{proj}_w}$  is used to project
the variable $w$ from the context. We postpone further discussion of the exact shape of the
clauses until Section~\ref{sec:refin}, since it depends on the representation we
use.
For the moment, we are just interested in composing the proof terms
occurring in these resolution steps into one composite proof term: 
$\kappa_{\elim} ~ \kappa_{\elimeqb} ~ \kappa_{\texttt{proj}_w} 
\kappa_{\texttt{subst}_{\A}}$.
%
%
Note that, by resolving the goal (\ref{ex:goal}), we obtain a  substitution $\theta$ that assigns the type $A$ to
the logic variable $?_B$, \ie $\theta(?_B) = A$.
%
At the same time, the proof term computed by the the derivation (\ref{ex:deriv})  is interpreted as a solution
$(\elimeqb ~ w)$ for the term-level
metavariable
$?_b$. However, 
the proof term can be used to reconstruct the derivation of well-typedness
of the judgement 
$m : \maybeA~\ttb , w : \ttb \eqb \ffb \vdash \elimeqb ~ w : \A$
as well.
In general,
a substitution is interpreted as a solution to a type-level metavariable and
a proof term as a solution to a term-level metavariable.
%
The remaining solution for $?_A$ is computed using similar methodology, and we omit the details here.
Thus, we have computed values for all metavariables in function (\ref{ex:tjust}),
\ie we inferred all types and synthesised all terms.

%
%

%

\subsection*{Contributions} 
Our main contributions are:
\begin{enumerate}
\item
  We present a novel
  approach to refinement
  for a first-order type
  theory with dependent types that is simpler than existing methods, \eg \citeN{PientkaD10}.

  \item We prove that generation of goals and logic programs from the extended language is decidable.

\item  We show that
  proof-relevant first-order Horn clause resolution 
  gives an appropriate inference mechanism for dependently typed languages: firstly, it is sound with respect to type checking in LF; secondly,
  the proof term construction alongside the resolution trace allows to reconstruct the derivation of well-typedness judgement.  

\end{enumerate}

\noindent
This paper is structured as follows.
Section~\ref{sec:nlfodtt} gives a nameless formulation of LF, the chosen first-order dependent
type theory. 
We then present the described refinement algorithm by means of a formal calculus in Section~\ref{sec:refin} and show that it is decidable.
Section~\ref{sec:sound}
establishes interpretation of answer substitutions and proof terms as solutions
to refinement problems and states soundness of the interpretation.
Finally, in Sections~\ref{sec:relwork} and \ref{sec:conclusion}
we discuss related and future work and conclude.
%
%


%


\section{Nameless LF}
\label{sec:nlfodtt}

Standard expositions of a type theory use variable names. However, variable
names carry a burden when implementing such a type theory. For
example, types need to be checked up to $\alpha$-equivalence of bound variables 
and fresh
names need to be introduced in order to expand terms to $\eta$-long form.
In order to avoid the burden, existing implementations use de Bruijn
indices. We use de Bruijn indices directly in our exposition as it allows us to
avoid the above problems when checking the equality of terms and types and when
synthesising new terms and types.

We use natural numbers in $\mathbb{N}$ for de Bruijn indices $\iota, \iota_1, \dots$, and we denote successor 
by $\sigma(-)$. 
We assume countably infinite disjoint sets $\mathcal{C}$ of
\emph{term constants}, and $\mathcal{B}$ of \emph{type constants}. 
We denote elements of $\mathcal{C}$ by  $c$, $c'$, \etc, and elements of
$\mathcal{B}$ by $\alpha$, $\beta$, \etc

\begin{definition}[Nameless LF]
    \label{def:nlf}
    The \emph{terms}, \emph{types}, and \emph{kinds} as
    well as \emph{signatures} and \emph{contexts} are: 

    \parbox{2.2in}{
    \begin{gather*}
    \begin{aligned}[c]
        \small
        \text{Kinds} &&  ~K & ::=~ && \text{type} 
            \mid \Pi T . K
            \\
        \text{Types} &&  T & ::= && \mathcal{B} \mid T t
            \mid \Pi T . T\\
        \text{Terms} &&  t & ::= && \mathcal{C} 
            \mid \mathbb{N}
            \mid \lambda T . t \mid t t
    \end{aligned}
    \end{gather*}
    }
    \parbox{2.5in}{
    \begin{gather*}
    \begin{aligned}[c]
        \text{Signatures} &&  ~Sig & ::=~ && \cdot
            \mid Sig , \mathcal{C} : T \mid
            Sig , \mathcal{B} : K
            \\
        \text{Contexts} && ~Con & ::= && \cdot
            \mid Con , T
            \\~\\
    \end{aligned}
    \end{gather*}
    }
\end{definition}

\noindent
We use identifier $L$ to denote kinds in $K$, 
identifiers $A$, $B$ to denote types in $T$ and identifiers $M$, $N$ to
denote terms in $t$. We use $\mathcal{S}$ for signatures and $\Gamma$ for
contexts.
In line with standard practice, we define two operations.
\emph{Shifting} recursively
traverses a term, a type, or a kind and increases all indices by one.

\begin{definition}[Shifting]
    Term and type shifting, denoted by $\shiftn{(-)}{\iota}$ is defined as
    follows:

\noindent
\parbox{0.33\textwidth}{
    \begin{gather*}
        \begin{aligned}[c]
            \shiftn{c}{\iota} & = c\\
            \shiftn{(\lambda A . M)}{\iota} & =
                \lambda \shiftn{A}{\iota} . \shiftn{M}{\sigma \iota}\\
            \shiftn{(M N)}{\iota} & = 
                \shiftn{M}{\iota} . \shiftn{N}{\iota}\\
        \end{aligned}
    \end{gather*}
}
\parbox{0.33\textwidth}{
    \begin{gather*}
        \begin{aligned}[c]
            \shiftn{\iota}{0} & = \sigma \iota\\
            \shiftn{0}{\sigma \iota} & = 0\\
            \shiftn{\sigma \iota}{\sigma \iota'} & =
                \sigma (\shiftn{\iota}{\iota'})
        \end{aligned}
    \end{gather*}
}
\parbox{0.33\textwidth}{
    \begin{gather*}
        \begin{aligned}[c]
            \shiftn{\alpha}{\iota} & = \alpha\\
            \shiftn{(\Pi A . B)}{\iota} & =
                \lambda \shiftn{A}{\iota} . \shiftn{B}{\sigma \iota}\\
            \shiftn{(A M)}{\iota} & = 
                \shiftn{A}{\iota} . \shiftn{M}{\iota}\\
        \end{aligned}
    \end{gather*}
}
\end{definition}

\noindent
\emph{Substitution} with a term $N$ and index $\iota$
replaces 
indices that are bound by the $\iota$-th binder while updating remaining indices.
The index $\iota$ is increased when traversing under a binder.

\begin{definition}[Substitution]
    Term and type substitution, denoted by $\substn{(-)}{N}{\iota}$ is defined as
    follows:

\noindent
\parbox{0.385\textwidth}{
    \begin{gather*}
        \begin{aligned}[c]
            \substn{c}{N}{\iota} & = c\\
            \substn{(\lambda A . M)}{N}{\iota} & =
                \lambda \substn{A}{N}{\iota} . 
                    \substn{M}{\shiftn{N}{0}}{\sigma \iota}\\
            \substn{(M_1 M_2)}{N}{\iota} & = 
                \substn{M_1}{N}{\iota} . \substn{M_2}{N}{\iota}\\
        \end{aligned}
    \end{gather*}
}
\parbox{0.22\textwidth}{
    \begin{gather*}
        \begin{aligned}[c]
            \substn{0}{N}{0} & = N\\
            \substn{0}{N}{\sigma \iota} & = 0\\
            \substn{\sigma \iota}{N}{0} & = \sigma \iota\\
            \substn{\sigma \iota}{N}{\sigma \iota'} & =
                \sigma (\substn{\iota}{N}{\iota'})
        \end{aligned}
    \end{gather*}
}
\parbox{0.2\textwidth}{
    \begin{gather*}
        \begin{aligned}[c]
            \substn{\alpha}{N}{\iota} & = \alpha\\
            \substn{(\Pi A . B)}{N}{\iota} & =
                \lambda \substn{A}{N}{\iota} . 
                    \substn{\shiftn{B}{0}}{N}{\sigma \iota}\\
            \substn{(A M)}{N}{\iota} & = 
                \substn{A}{N}{\iota} . \substn{M}{N}{\iota}\\
        \end{aligned}
    \end{gather*}
}
\end{definition}

\noindent
Shifting with a greater index than $0$ and 
substitution for other indices than $0$ is not necessary
in the inference rules of neither the internal language nor refinement.
For the sake of readability we introduce the following abbreviations:

\noindent
\parbox{0.25\textwidth}{
\begin{gather*}
    \shift{A} \overset{def}{=} \shiftn{A}{0}\\
\end{gather*}
}
\parbox{0.25\textwidth}{
\begin{gather*}
    \shift{M} \overset{def}{=} \shiftn{M}{0}\\
\end{gather*}
}
\parbox{0.25\textwidth}{
\begin{gather*}
    \subst{A}{N} \overset{def}{=} \substn{A}{ N }{ 0}\\
\end{gather*}
}
\parbox{0.25\textwidth}{
\begin{gather*}
    \subst{M}{N} \overset{def}{=} \substn{M}{ N }{ 0}\\
\end{gather*}
}

\vskip-1.5em
\noindent
Well-formedness of objects introduced by
Definition~\ref{def:nlf} is stated by means of several judgements.
In particular, we give equality in nameless LF as algorithmic, following Harper
and Pfenning~\cite{HarperP05}. 
In order to do so we define 
 simple kinds, simple types, simple signatures, and simple contexts:

\begin{definition}
    The \emph{simple kinds}, \emph{simple types}, and 
     \emph{simple signatures} are:
    \begin{gather*}
    \begin{aligned}[c]
        \text{Simple kinds} &&  ~K^- & ::=~ && \typeK^- 
            \mid T^- \rightarrow K^-
            \\
        \text{Simple types} &&  T^- & ::= && \mathcal{B} 
            \mid T^- \rightarrow T^- \\
        \text{Simple signatures} &&  ~Sig^- & ::=~ && \cdot
            \mid Sig^- , \mathcal{C} : T^-
            \mid Sig^- , \mathcal{B} : K^-
            \\
        \text{Simple contexts} && ~Con^- & ::= && \cdot
            \mid Con^- , T^-
            \\
    \end{aligned}
    \end{gather*}
\end{definition}
We use identifiers $\kappa$ for simple kinds, $\tau$ for simple types, 
$\sig^-$ for simple signatures and $\Delta$ for simple contexts. The
erasure from objects to corresponding simple objects, denoted $(-)^-$
is defined as follows:

\begin{definition}[Erasure]

\noindent
\parbox{0.5\textwidth}{
\begin{gather*}
    \begin{aligned}[c]
        (\typeK)^- &= \typeK\\
        (\Pi A . L)^- &= (A)^- \to (L)^-
    \end{aligned}
\end{gather*}
}
\parbox{0.5\textwidth}{
\begin{gather*}
    \begin{aligned}[c]
        (\alpha)^- &= \alpha\\
        (\Pi A . B)^- &= (A)^- \to (B)^-\\
        (A M)^- &= (A)^- 
    \end{aligned}
\end{gather*}
}
\end{definition}

The well-formedness of judgements for kinds, types and terms, 
weak algorithmic equality of types, algorithmic and structural equality of
terms, and weak head reduction of terms are:

\noindent
\parbox{0.25\textwidth}{
    \begin{gather*}
        \mathcal{S} ; \Gamma \vdash L : \kindS\\
        \mathcal{S}^- ; \Delta \vdash A \rightleftharpoons A' : \kappa
    \end{gather*}
}
\parbox{0.25\textwidth}{
    \begin{gather*}
        \mathcal{S} ; \Gamma \vdash A : L\\
        \mathcal{S} ; \Gamma \vdash M \leftrightarrow M' : \tau
    \end{gather*}
}
\parbox{0.25\textwidth}{
    \begin{gather*}
        \mathcal{S} ; \Gamma \vdash M : A\\
        \mathcal{S} ; \Gamma \vdash M \Leftrightarrow M' : \tau
    \end{gather*}
}
\parbox{0.25\textwidth}{
    \begin{gather*}
        M \whr M'
    \end{gather*}
}

\begin{figure*}
    \framebox{$
        \sig   ;  \Gamma  \vdash  L  :  \kindS
    $}\raggedright
    \vskip1em
    \noindent
    \parbox{1.6in}
    {
        \centering
        \AxiomC{$
            \sig  \vdash \Gamma \text{ ctx} 
        $}
        \RightLabel{\sc K-ty}
        \UnaryInfC{$
            \sig ; \Gamma \vdash \typeK : \kindS
        $}
        \DisplayProof
    }
%
    \parbox{3.6in}
    {
        \centering
        \AxiomC{$
            \sig ; \Gamma \vdash A : \typeK
        $}
        \AxiomC{$
            \sig ; \Gamma, A \vdash L : \kindS
        $}
        \RightLabel{\sc K-$\Pi$-intro}
        \BinaryInfC{$
            \sig ; \Gamma \vdash \Pi A . L : \kindS
        $}
        \DisplayProof
    }

    \caption{Well-formedness of nameless kinds}
    \label{fig:wfnlkinds}
\end{figure*}

\begin{figure*}

    \framebox{$
        \mathcal{S}  ;  \Gamma  \vdash A : L
    $}\raggedright
    \vskip1em

    \parbox{1.9in}
    {
        \centering
        \AxiomC{$
            \sig \vdash \Gamma \text{ ctx}
        $}
        \AxiomC{$
            \alpha : L \in \mathcal{S}
        $}
        \RightLabel{\sc T-con}
        \BinaryInfC{$
            \sig ; \Gamma \vdash  \alpha : L
        $}
        \DisplayProof
    }
    \parbox{3.3in}
    {
        \centering
        \AxiomC{$
            \sig ; \Gamma \vdash A : \typeK
        $}
        \AxiomC{$
            \sig ; \Gamma, A \vdash B : \typeK
        $}
        \RightLabel{\sc T-$\Pi$-intro}
        \BinaryInfC{$
            \sig ; \Gamma \vdash \Pi A . B : \typeK
        $}
        \DisplayProof
    }
    \vskip1em
    \parbox{5.2in}
    {
        \centering
        \AxiomC{$
            \sig ; \Gamma \vdash A : \Pi B . L
        $}
        \AxiomC{$
            \sig ; \Gamma \vdash M : B'
        $}
        \AxiomC{$
            \sig^- ; \Gamma^- \vdash B \rightleftharpoons B' : \typeK^-
        $}
        \RightLabel{\sc T-$\Pi$-elim}
        \TrinaryInfC{$
            \sig ; \Gamma \vdash A M : \subst{L}{M}
        $}
        \DisplayProof
    }

    \caption{Well-formedness of nameless types}
    \label{fig:wfnltypes}
\end{figure*}

\begin{figure*}[t]

    \framebox{$
        \mathcal{S}  ;  \Gamma  \vdash M : A
    $}\raggedright
    \vskip1em

    \parbox{2.0in}
    {
        \centering
        \AxiomC{$
            \sig \vdash \Gamma \text{ ctx}
        $}
        \AxiomC{$
            c : A \in \sig \
        $}
        \RightLabel{\sc con}
        \BinaryInfC{$
            \sig ; \Gamma \vdash c : A
        $}
        \DisplayProof
    }
    \parbox{1.7in}
    {
        \centering
        \AxiomC{$
            \sig \vdash \Gamma,A \text{ ctx}
        $}
        \RightLabel{\sc zero}
        \UnaryInfC{$
            \sig ; \Gamma, A \vdash 0 : \shift{A} 
        $}
        \DisplayProof
    }
    \parbox{1.5in}
    {
        \centering
        \AxiomC{$
            \sig ; \Gamma \vdash \iota : A
        $}
        \RightLabel{\sc succ}
        \UnaryInfC{$
            \sig ; \Gamma, B \vdash \sigma \iota  : \shift{A} 
        $}
        \DisplayProof
    }
    \vskip1em
    \parbox{4.8in}
    {
        \centering
        \AxiomC{$
            \sig ; \Gamma \vdash A : \typeK
        $}
        \AxiomC{$
            \sig ; \Gamma,A \vdash M : B
        $}
        \RightLabel{\sc $\Pi$-intro}
        \BinaryInfC{$
            \sig ; \Gamma \vdash \lambda A . M : \Pi A . B
        $}
        \DisplayProof
    }
    \vskip1em
    \parbox{4.8in}
    {
        \centering
        \AxiomC{$
            \sig ; \Gamma \vdash M : \Pi A . B
        $}
        \AxiomC{$
            \sig ; \Gamma \vdash N : A'
        $}
        \AxiomC{$
            \sig^- ; \Gamma^- \vdash A \rightleftharpoons A' : \typeK
        $}
        \RightLabel{\sc $\Pi$-elim}
        \TrinaryInfC{$
            \sig ; \Gamma \vdash M N : \subst{B}{N}
        $}
        \DisplayProof
    }

    \caption{Well-formedness of nameless terms}
    \label{fig:wfnlterms}
\end{figure*}

\noindent
The inference rules for well-formedness of kinds, types, and terms
are listed in Figures~\ref{fig:wfnlkinds},
\ref{fig:wfnltypes}, and~\ref{fig:wfnlterms}. The inference rules of
well-formedness judgements for signatures and contexts as well as
for definitional equality are standard (\cf \citeN{HarperP05}).
The inference rules for weak head reduction are listed in Figure \ref{fig:whr}.
The well-formedness of signatures and contexts is defined in Figure 
\ref{fig:wfsigcon}. Algorithmic equality of terms, structural equality of terms
and weak algorithmic equality of types are defined in Figures
\ref{fig:eqtea}, \ref{fig:eqtes}, and \ref{fig:eqty} respectively.

\begin{figure}
    \framebox{$
        M \whr M'
    $}\raggedright
    \vskip1em
    \noindent
    \parbox{0.4\textwidth}
    {
        \centering
        \AxiomC{}
        \UnaryInfC{$
            (\lambda A . M) N \whr \subst{M}{N}
        $}
        \DisplayProof
    }
    \parbox{0.59\textwidth}
    {
        \centering
        \AxiomC{$
            M \whr M'
        $}
        \UnaryInfC{$
            M N \whr M' N'
        $}
        \DisplayProof
    }

    \caption{Weak head reduction of terms}
    \label{fig:whr}
\end{figure}

\begin{figure}
    \framebox{$
        \vdash \sig \text{ sig}
    $}\raggedright
    \vskip1em
    \noindent
    \parbox{0.10\textwidth}
    {
        \centering
        \AxiomC{}
        \UnaryInfC{$
            \vdash
            \cdot \text{ sig}
        $}
        \DisplayProof
    }
    \parbox{0.43\textwidth}
    {
        \centering
        \AxiomC{$
            \vdash \sig \text{ sig}
        $}
        \AxiomC{$
            \sig ; \cdot 
            \vdash L : \kindS
        $}
        \AxiomC{$
            a \not\in \sig
        $}
        \TrinaryInfC{$
            \vdash \sig, a : L \text{ sig}
        $}
        \DisplayProof
    }
    \parbox{0.42\textwidth}
    {
        \centering
        \AxiomC{$
            \vdash \sig \text{ sig}
        $}
        \AxiomC{$
            \sig ; \cdot 
            \vdash A : \typeK
        $}
        \AxiomC{$
            c \not\in \sig
        $}
        \TrinaryInfC{$
            \vdash \sig, c : A \text{ sig}
        $}
        \DisplayProof
    }

    \vskip1em
    \framebox{$
        \sig \vdash \Gamma \text{ ctx}
    $}\raggedright
    \vskip1em
    \noindent
    \parbox{0.4\textwidth}
    {
        \centering
        \AxiomC{$
        \vdash \sig \text{ sig}
        $}
        \UnaryInfC{$
            \sig \vdash
            \cdot \text{ ctx}
        $}
        \DisplayProof
    }
    \parbox{0.59\textwidth}
    {
        \centering
        \AxiomC{$
            \sig \vdash \Gamma \text{ ctx}
        $}
        \AxiomC{$
            \sig  ; \Gamma
            \vdash A : \typeK
        $}
        \BinaryInfC{$
            \vdash \sig; \Gamma , A \text{ ctx}
        $}
        \DisplayProof
    }

    \caption{Well-formedness of signatures and contexts}
    \label{fig:wfsigcon}
\end{figure}

\begin{figure}
    \framebox{$
        \sig^- ; \Delta \vdash M \Leftrightarrow M' : \tau
    $}\raggedright
    \vskip1em

    \noindent
    \parbox{0.49\textwidth}
    {
        \centering
        \AxiomC{$
            M \whr M'
        $}
        \AxiomC{$
            \sig^- ; \Delta 
            \vdash
            M' \Leftrightarrow N
        $}
        \BinaryInfC{$
            \sig^- ; \Delta 
            \vdash
            M \Leftrightarrow N
        $}
        \DisplayProof
    }
    \parbox{0.49\textwidth}
    {
        \AxiomC{$
            N \whr N'
        $}
        \AxiomC{$
            \sig^- ; \Delta 
            \vdash
            M \Leftrightarrow N'
        $}
        \BinaryInfC{$
            \sig^- ; \Delta 
            \vdash
            M \Leftrightarrow N
        $}
        \DisplayProof
    }
    \vskip1em

    \noindent
    \parbox{0.49\textwidth}
    {
        \AxiomC{$
            \sig^- ; \Delta 
            \vdash
            M \leftrightarrow N : \tau
        $}
        \UnaryInfC{$
            \sig^- ; \Delta 
            \vdash
            M \Leftrightarrow N : \tau
        $}
        \DisplayProof
    }
    \parbox{0.49\textwidth}
    {
        \centering
        \AxiomC{$
            \sig^- ; \Delta, \tau_1
            \vdash 
            (\shift{M})0 \Leftrightarrow (\shift{N})0 : \tau_2
        $}
        \UnaryInfC{$
            \sig^- ; \Delta
            \vdash 
            M \Leftrightarrow N : \tau_1 \to \tau_2
        $}
        \DisplayProof
    }

    \caption{Algorithmic equality of terms}
    \label{fig:eqtea}
\end{figure}

\begin{figure}
    \vskip1em
    \framebox{$
        \sig^- ; \Delta \vdash M \leftrightarrow N : \tau
    $}\raggedright
    \vskip1em
    \noindent
    \parbox{0.33\textwidth}
    {
        \centering
        \AxiomC{$
        \vdash \sig^- \text{ ssig}
        $}
        \UnaryInfC{$
            \sig^- ; \Delta , \tau
            \vdash
            0 \leftrightarrow 0 : \tau
        $}
        \DisplayProof
    }
    \parbox{0.33\textwidth}
    {
        \centering
        \AxiomC{$
            \sig^- ; \Delta \vdash \iota \leftrightarrow \iota' : \tau
        $}
        \UnaryInfC{$
            \sig^- ; \Delta, \tau' 
            \vdash 
            \sigma \iota \leftrightarrow \sigma \iota' : \tau
        $}
        \DisplayProof
    }
    \parbox{0.32\textwidth}
    {
        \centering
        \AxiomC{$
        \vdash \sig^- \text{ ssig}
        $}
        \AxiomC{$
        c : \tau \in \sig^-
        $}
        \BinaryInfC{$
            \sig^- ; \Delta 
            \vdash
            c \leftrightarrow c : \tau
        $}
        \DisplayProof
    }

    \vskip1em
    \parbox{\textwidth}
    {
        \centering
        \AxiomC{$
            \sig^- ; \Delta
            \vdash 
            M_1 \leftrightarrow N_1 : \tau_2 \to \tau_1
        $}
        \AxiomC{$
            \sig^- ; \Delta
            \vdash 
            M_2 \Leftrightarrow N_2 : \tau_2
        $}
        \BinaryInfC{$
            \sig^- ; \Delta 
            \vdash
            M_1 M_2 \leftrightarrow N_1 N_2 : \tau_1
        $}
        \DisplayProof
    }
    \caption{Structural equality of terms}
    \label{fig:eqtes}
\end{figure}

\begin{figure}
    \framebox{$
        \sig^- ; \Delta \vdash A \rightleftharpoons A' : \kappa
    $}\raggedright
    \vskip1em

    \noindent
    \parbox{0.33\textwidth}
    {
        \centering
        \AxiomC{$
            \vdash \sig^- \text{ ssig}
        $}
        \AxiomC{$
            \alpha : \kappa \in \sig^- 
        $}
        \BinaryInfC{$
            \sig^- ; \Delta
            \vdash 
            \alpha \rightleftharpoons \alpha : \kappa
        $}
        \DisplayProof
    }
    \parbox{0.63\textwidth}
    {
        \centering
        \AxiomC{$
            \sig^- ; \Delta
            \vdash 
            A \rightleftharpoons B : \tau \to \kappa
        $}
        \AxiomC{$
            \sig^- ; \Delta
            \vdash 
            M \Leftrightarrow N : \tau
        $}
        \BinaryInfC{$
            \sig^- ; \Delta 
            \vdash
            A M \rightleftharpoons B N : \kappa
        $}
        \DisplayProof
    }

    \vskip1em
    \parbox{\textwidth}
    {
        \centering
        \AxiomC{$
            \sig^- ; \Delta
            \vdash
            A_1 \rightleftharpoons B_1 : \typeK
        $}
        \AxiomC{$
            \sig^- ; \Delta , (A_1)^-
            \vdash
            (\shift{A_2}) {}\rightleftharpoons (\shift{B_2}) : \typeK
        $}
        \BinaryInfC{$
            \sig^- ; \Delta 
            \vdash
            (\Pi A_1 . A_2) \rightleftharpoons (\Pi B_1 . B_2) : \typeK
        $}
        \DisplayProof
    }
    \caption{Weak algorithmic equality of types}
    \label{fig:eqty}
\end{figure}

\section{Refinement in Nameless LF}
\label{sec:refin}

Following the ideas we sketched in Section~\ref{sect:intro}, we present
the translation of a refinement problem into Horn clause logic with explicit
proof terms.
Firstly, we extend the language of nameless LF with metavariables, which allows us
to capture incomplete terms. Secondly, we describe the language of
Horn clause logic with explicit proof terms. Finally, we give a calculus for
transformation of an incomplete term to a goal and a program. 

\subsection{Refinement Problem}

We capture missing information in nameless LF terms by metavariables. 
We assume infinitely countable disjoint sets  $?_\mathcal{B}$ and
$?_\mathcal{V}$ that stand for omitted types and terms and we call elements of
these sets type-level and term-level metavariables respectively.
We use identifiers $?_a$, $?_b$, \etc to denote elements of $?_\mathcal{V}$ and
identifiers $?_A$, $?_B$, \etc to denote elements of
$?_\mathcal{B}$. The extended syntax is defined as follows:
\begin{definition}[Extended Nameless LF]
    We define \emph{extended nameless types, terms and contexts} as follows:

\noindent
\parbox{2.5in}{
    \begin{gather*}
    \begin{aligned}[c]
        \text{Types} &&  T & ::= && \dots \mid{} ?_\mathcal{B} \\
        \text{Terms} &&  t & ::= && \dots \mid{} ?_\mathcal{V} \\
    \end{aligned}
    \end{gather*}
}
\parbox{2.5in}{
    \begin{gather*}
    \begin{aligned}[c]
        \text{Contexts} && ~Con & ::= && \dots \mid Con , ?_\mathcal{V} : T
    \end{aligned}
    \end{gather*}
}
\end{definition}
%
The ellipsis in the definition are to be understood as the appropriate syntactic
constructs of Definition \ref{def:nlf}. Note that we do not define an extended
signature. We assume that the signature is always fixed and does not contain any
metavariables. This does not pose any problem since well-typedness of signature
does not depend on the term being refined.
We use $\mtvar(-)$ and $\mvar(-)$ to denote the sets of type-level 
and term-level metavariables respectively. 
The well-formedness judgements of the
nameless LF are then defined on a subset of extended 
objects.
%

\begin{lemma}
  \label{lem:ground}
    Let $L$ be an extended nameless kind, $A$ an extended nameless type and $M$
    an extended nameless term. 
    Let $\sig$ and $\Gamma$ be contexts.
    \begin{itemize}
        \item If $\sig ; \Gamma \vdash L : \kindS$ then $\mvar(L) = \emptyset$ and
            $\mtvar(L) = \emptyset$,
        \item if $\sig ; \Gamma \vdash A : L$ then $\mvar(A) = \emptyset$ and
            $\mtvar(A) = \emptyset$, and
        \item if $\sig ; \Gamma \vdash M : A$ then $\mvar(M) = \emptyset$ and
            $\mtvar(M) = \emptyset$.
    \end{itemize}
  \end{lemma}
\begin{proof}
    By induction on the derivation of judgements.
\end{proof}

\noindent
A \emph{refinement problem} is defined as a term in the extended syntax.
A signature and a context of the term are kept implicit.

\begin{example}[Refinement Problem]\label{ex:rp}
Taking the problem from Introduction, the term $M'$ given by
$ (\elimmaybeA~\ttb~0)
(\lambda ?_A . ?_b)
$ is a refinement problem.  
The appropriate context is 
$
\Gamma_1 = \cdot , \maybeA~\ttb
$. This signature in Introduction is adjusted to nameless signature $\sig$.
  \end{example}

 A
\emph{refinement} of 
a term
is a pair of assignments
$(\rho, R)$ such that 
$\rho : ~?_\mathcal{V} \to t$ is an assignment of (extended) terms to term-level
metavariables and
$R : ~?_\mathcal{B} \to T$ is an assignment of (extended) types to type-level
metavariables.
We define application of refinement $(\rho,R)(-)$ to terms, types and kinds by
induction on definition of the syntactic object. 

\pagebreak
\begin{definition}[Refinement application]
    Let $\rho : ~?_\mathcal{V} \to t$ be an assignment of terms and $R :
    ~?_\mathcal{B} \to T$ be an assignment of types. Application of the
    refinement $(\rho, R)$ to kinds, types and terms is defined by:\\
    \parbox{0.5\textwidth}{
    \begin{gather*}
        \begin{aligned}[c]
            (\rho, R)(\typeK) &= \typeK\\
            (\rho, R)(\Pi A . L) &= \Pi (\rho, R)(A) . (\rho, R)(L)\\
            (\rho, R)(\alpha) &= \alpha \\
            (\rho, R)(?_A) &= R(?_A) \\
            (\rho, R)(\Pi A . B) &= \Pi (\rho, R)(A) . (\rho, R)(B)\\
            (\rho, R)(A N) &= (\rho, R)(A) (\rho, R)(N)\\
        \end{aligned}
    \end{gather*}
    }
    \parbox{0.5\textwidth}{
    \begin{gather*}
        \begin{aligned}[c]
            (\rho, R)(c) &= c \\
            (\rho, R)(\iota) &= \iota \\
            (\rho, R)(?_a) &= \rho(?_a) \\
            (\rho, R)(\lambda x : A . M) &= \lambda x : (\rho, R)(A) . (\rho, R)(M)\\
            (\rho, R)(M N) &= (\rho, R)(M) (\rho, R)(N)\\
        \end{aligned}
    \end{gather*}
}

\end{definition}

\noindent
A \emph{solution} to
a refinement problem $t$ is a refinement $(\rho, R)$ such that $(\rho, R)(t)$ is
a well-formed term of nameless LF. That is, by Lemma~\ref{lem:ground}, $(\rho, R)(t)$
does not contain neither term- nor type-level metavariables.


Horn clause logic is usually presented using a signature that comprises sets of
function and predicate symbols and the appropriate grammar for atomic and Horn
formulae (\cf \citeN{MillerN12}). Although the same presentation can be given for
Horn clause logic with explicit proof terms, for the sake of brevity we resort
to a simpler presentation that is sufficient for our purposes.
We define atomic formulae using objects of
nameless LF and we list all the predicates that are needed for refinement
translation explicitly in the grammar. However, we make an exception in the case of
contexts and use the usual list notation.
Metavariables of extended nameless LF
are seen as logic variables.
Furthermore, we assume a finite set $\mathcal{K}$
of \emph{proof-term symbols}
and a countable set $\mathcal{D}$ of goal variables.  We denote elements of
$\mathcal{K}$ by $\kappa$ with indices and elements of $\mathcal{D}$ by
$\gamma$ with indices. For technical reasons, we also use metavariables in
positions of kinds, denoted $?_L$, and in position of indicies,
denoted $?_\iota$.
The syntax is defined as follows:
\begin{definition}[Syntax of Horn Clause Logic with Explicit Proof Terms]\label{def:HCl}
    \emph{Atomic formulae}, \emph{Horn formulae}, \emph{programs},
    \emph{proof terms}, and \emph{goals} are generated as follows:

    \parbox{3.2in}
    {
        \centering
        \begin{gather*}
        \begin{aligned}[c]
            \text{Atomic formulae} && At & ::= && 
            eq_{t}^a(t, t, T, Con) \mid
            eq_{t}^s(t, t, T, Con) \mid\\
            && & &&
            eq_{T}(T, T, K, Con) \mid
            eq_{K}(K, K, Con) \mid\\
            && & &&
            type(T, K, Con) \mid
            term(t, T, Con) \mid
            \shifteq{T}{T} \mid
            \substeq{T}{t}{T'} \mid
            \top 
        \end{aligned}
        \end{gather*}
    }

    \parbox{2.7in}
    {
        \centering
        \begin{gather*}
        \begin{aligned}[c]
            \text{Horn clauses} && HC & ::= && At \leftarrow At \land \ldots \land At\\
            \text{Proof terms} &&  ~PT & ::=~ && \mathcal{K}~PT \ldots PT
        \end{aligned}
        \end{gather*}
    }
    \parbox{2.5in}
    {
        \centering
        \begin{gather*}
        \begin{aligned}[c]
            \text{Programs} && ~\mathcal{P} & ::= && \cdot \mid P , \mathcal{K} : HC\\
            \text{Goals} && ~\mathcal{G} & ::= && \mathcal{D} : At \land \ldots \land 
                \mathcal{D} : At
        \end{aligned}
        \end{gather*}
    }
\end{definition}
The atomic formula $\top$ is intended to stand for a
formula that is always true.
The predicates $eq_t^a$ and $eq_t^s$
denote algorithmic and structural equality respectively of
terms of a certain simple type in a context, 
the predicates $eq_T$ and $eq_K$
denote equality of
terms of a certain simple kind, and
equality of kinds in a context respectively.
The predicates $\term$ and
$\type$ 
denote,  respectively, that a term or a type is well-formed in a context.
We use $\shifteq{A}{A'}$ to denote that a type $A'$ is the result of
shifting of $A$; 
and we use
$\substeq{A}{M}{A'}$ to denote that $A'$ is the result of substitution of $A$
with $M$.
We use the identifier $H$ to denote Horn clauses in \emph{HC}.
Goals in $\mathcal{G}$ are denoted by $G$. 
We use $P$ with indices to refer to programs. 
Proof terms in \emph{PT} are denoted by $\delta$, $\delta_1$, \etc

\begin{figure*}

    \framebox{$
        \sig   ;  \Gamma ; M  \vdash  ( G \mid A )
    $}\raggedright

    \vskip1em

    \parbox{2.2in}
    {
        \centering
        \AxiomC{$
            c : A \in \sig
        $}
        \RightLabel{\sc r-con}
        \UnaryInfC{$
            \sig ; \Gamma ; c \vdash (\top \mid A)
        $}
        \DisplayProof
    }
    \parbox{3.0in}
    {
        \centering
        \AxiomC{$
            ~
        $}
        \RightLabel{\sc r-t-meta}
        \UnaryInfC{$
            \sig ; \Gamma ; ?_a \vdash (?_a : \term(?_{a'}, ?_A, \Gamma) \mid \; ?_A)
        $}
        \DisplayProof
    }

    \vskip1em

    \parbox{2.5in}
    {
        \centering
        \AxiomC{$
            ~
        $}
        \RightLabel{\sc r-zero}
        \UnaryInfC{$
            \sig ; \Gamma, A ; 0 \vdash ( \shifteq{A}{?_A} \mid ?_A)
        $}
        \DisplayProof
    }
    \parbox{2.7in}
    {
        \centering
        \AxiomC{$
            \sig ; \Gamma ; \iota \vdash (G \mid A)
        $}
        \RightLabel{\sc r-succ}
        \UnaryInfC{$
            \sig ; \Gamma, B ; \sigma \iota
            \vdash
            (G \land ( \shifteq{A}{?_A} ) \mid \; ?_A)
        $}
        \DisplayProof
    }

    \vskip1em

    \parbox{5.2in}
    {
        \centering
        \AxiomC{$
            \sig ; \Gamma ; A
            \vdash
            ( G_A \mid L )
        $}
        \AxiomC{$
            \sig ; \Gamma, A ; M \vdash ( G_M \mid B )
        $}
        \RightLabel{\sc r-$\lambda$-intro}
        \BinaryInfC{$
            \sig ; \Gamma ; \lambda A . M
            \vdash
            (G_A \land G_M
            \land eq_K(L, \typeK, \Gamma)
            \mid \; \Pi A . B)
        $}
        \DisplayProof
    }

    \vskip1em

    \parbox{5.2in}
    {
        \centering
        \AxiomC{$
            \sig ; \Gamma ; M 
            \vdash
            ( G_M \mid A )
        $}
        \AxiomC{$
            \sig ; \Gamma ; N 
            \vdash
            ( G_N \mid A_2 )
        $}
        \RightLabel{\sc r-$\lambda$-elim}
        \BinaryInfC{$
            \sig ; \Gamma ; M N
            \vdash
            ( G_M \land G_N 
            \land eq_T(A, \Pi A_2 . ?_B, \typeK, \Gamma) 
            \land (\substeq{?_B}{N}{?_{B'}})
            \mid \; ?_{B'})
        $}
        \DisplayProof
    }

    \caption{Refinement of terms}
    \label{fig:refterms}
\end{figure*}

Note that, by Definition~\ref{def:HCl}, atoms in goals
are assigned variables. 
Later, proof terms computed by resolution are identified by
these variables (Definition~\ref{def:prr}).
We omit explicit mention of goal variables whenever we do not need to refer to 
proof terms that the variables identify
(fresh variables are assumed in such cases).  

\subsection{From a Refinement Problem to a  Logic Program}

In this section, we explain how a term with metavariables is transformed into a
goal, and
the signature 
into a logic program.
At the end of the section we state
that, for a refinement problem, either a goal and a program exist or else the problem
cannot be refined to a well-formed term.

 We define a calculus with two kinds of
judgements, one for transforming refinement problems into goals and the other -- for
transforming signatures into logic programs.
These judgements are defined mutually in a similar way to
the well-formedness judgements of nameless LF in Figures \ref{fig:wfnltypes} and
\ref{fig:wfnlterms}.
We use $\sig ; \Gamma ; M \vdash (G \mid A)$ to denote the transformation of a term
$M$ in a signature $\sig$ and a context $\Gamma$ to a goal $G$. The
judgement also synthesises a type $A$ of the term $M$. Similarly, 
$\sig ; \Gamma ; A \vdash (G \mid K)$ denotes a transformation of a type $A$ in
$\sig$ and $\Gamma$ to a goal $G$ while synthesising a kind $K$.

\begin{definition}[Refinement Goals]
    The judgements $\sig ; \Gamma; M \vdash (G \mid A)$ and
    $\sig ; \Gamma ; A \vdash ( G \mid L)$ 
    are given by inference rules in Figures
        \ref{fig:refterms} and 
        \ref{fig:reftypes}.
    Metavariables that do not occur among assumptions have an implicit
    freshness condition.
\end{definition}

\begin{figure*}

    \framebox{$
        \sig   ;  \Gamma ; A  \vdash  ( G \mid L )
    $}\raggedright

    \vskip1em

    \parbox{2.4in}
    {
        \centering
        \AxiomC{$
            a : L \in \sig
        $}
        \RightLabel{\sc r-tcon}
        \UnaryInfC{$
            \sig ; \Gamma ; a \vdash (\top \mid L)
        $}
        \DisplayProof
    }
    \parbox{2.8in}
    {
        \centering
        \AxiomC{$
            ~
        $}
        \RightLabel{\sc r-T-meta}
        \UnaryInfC{$
            \sig ; \Gamma ; ?_A \vdash (\type(?_A , ?_L, \Gamma) \mid ?_L)
        $}
        \DisplayProof
    }

    \vskip1em

    \parbox{5.2in}
    {
        \centering
        \AxiomC{$
            \sig ; \Gamma ; A \vdash ( G_A \mid L_1)
        $}
        \AxiomC{$
            \sig ; \Gamma , A ; B \vdash ( G_B \mid L_2 )
        $}
        \RightLabel{\sc r-$\Pi$-intro}
        \BinaryInfC{$
            \sig ; \Gamma ; \Pi A . B
            \vdash
            (G_A \land G_B
            \land eq_K(L_1, \typeK, \Gamma)
            \land eq_K(L_2, \typeK, \Gamma)
            \mid \typeK)
        $}
        \DisplayProof
    }

    \vskip1em

    \parbox{5.2in}
    {
        \centering
        \AxiomC{$
            \sig ; \Gamma ; A \vdash ( G_A \mid L )
        $}
        \AxiomC{$
            \sig ; \Gamma ; M 
            \vdash
            ( G_M \mid B )
        $}
        \RightLabel{\sc r-$\Pi$-elim}
        \BinaryInfC{$
            \sig ; \Gamma ; A M
            \vdash
            ( G_A \land G_M \land eq_K(L, \Pi B . ?_L , \Gamma) 
            \land (\substeq{?_L}{M}{?_{L'}}) \mid \; ?_{L'})
        $}
        \DisplayProof
    }

    \caption{Refinement of types}
    \label{fig:reftypes}
\end{figure*}

The inference judgement for a logic program generation is denoted by $\sig
\vdash_{\text{Prog}} P$ where $\sig$ is a signature and $P$ is a generated logic program.
A generated logic program contains clauses that represent inference rules of type
theory and clauses that are generated from a signature $\mathcal{S}$. 
The clauses that represent inference rules of LF are the same for all programs and Definition~\ref{def:Pempty}
gives a minimal program $P_{e}$ that contains only these clauses. 
\begin{definition}
    \label{def:Pempty}
    Let  $P_{e}$ be a program with clauses that represent inference rules for
    well-formedness of terms and types:
    \begin{gather*}
        \begin{aligned}[clr]
            \kappa_{true} : \quad & \top \leftarrow  
                && \\
            \kappa_{0} : \quad & proj(0, ?_A, ?_A' : \; ?_\Gamma) \leftarrow  
                \; \shifteq{?_{A'}}{?_A}
                && \\
            \kappa_{\sigma} : \quad & proj(\sigma(?_\iota), ?_A, ?_B : \; ?_\Gamma) \leftarrow  
                proj(?_\iota, ?_A', ?_{\Gamma})
                \land
                \shifteq{?_{A'}}{ ?_A }
                && \\
            \kappa_{proj} : \quad & \type(?_\iota, ?_A, \typeK, ?_\Gamma) \leftarrow  
                            proj(?_\iota, ?_A, ?_\Gamma) \\
            \kappa_{\text{T-elim}} : \quad & \type(?_{A} ?_{M}, ?_{L}, ?_\Gamma) \leftarrow  
                \type(?_A , \Pi ?_{A_1}.?_{L'}, ?_\Gamma) \land
                \term(?_M, ?_{A_2}, ?_\Gamma) \land 
                \\ & \quad \quad 
                eq_T(?_{A_1},?_{A_2},\typeK, ?_\Gamma) 
                \land \substeq{?_{L'}}{?_M}{?_{L}}
                \\
            \kappa_{\text{T-intro}} : \quad & \type( \Pi ?_A . ?_B , \typeK, ?_\Gamma) \leftarrow
                \type(?_A, \typeK, ?_\Gamma) 
                \land
                \type(?_B, \typeK, ?_B : \;  ?_{\Gamma}) 
                \\
            \kappa_{\text{t-elim}} : \quad & \term(?_{M} ?_{N}, ?_{B}, ?_\Gamma) \leftarrow  
                \term(?_M , \Pi ?_{A_1}.?_{B'}, ?_\Gamma) \land
                \term(?_N, ?_{A_2}, ?_\Gamma) \land 
                \\ & \quad \quad 
                eq_T(?_{A_1},?_{A_2},\typeK, ?_\Gamma) 
                \land \substeq{?_{B'}}{?_N}{?_{B}} 
                \\
            \kappa_{\text{t-intro}} : \quad & \term( \lambda ?_A . ?_M , \Pi
                ?_A . ?_B, ?_\Gamma) \leftarrow
                \type(?_A, \typeK, ?_\Gamma) 
                \land
                \term(?_M, ?_{B}, ?_{\Gamma}) 
        \end{aligned}
    \end{gather*}

    \noindent Further, there are clauses that represent weak algorithmic equality
    of types, algorithmic and structural equality of termss, and weak head
    reduction of terms:
        \begin{gather*}
            \begin{aligned}[clr]
                \kappa_{\text{eqTintro}} : \quad & eq_T(\Pi ?_{A_1} . ?_{A_2}, 
                    \Pi ?_{B_1} . ?_{B_2} , \type, ?_\Gamma) \leftarrow  
                        eq_T(?_{A_1}, ?_{B_1}, \type, ?_\Gamma)
                        \land
                    \\ & \quad \quad
                    eq_T(?_{A_2}, ?_{B_2}, \type, ?_{A_1} : \; ?_\Gamma)
                    && \\
                \kappa_{\text{eqTelim}} : \quad & eq_T(?_A ?_M , ?_B ?_N , ?_L,
                ?_\Gamma) \leftarrow  
                    eq_T(?_A, ?_B, \Pi ?_C . ?_L, ?_{\Gamma}) \land
                    eq^a_t(?_M, ?_N, ?_C, ?_{\Gamma})
                    && \\
                \kappa_{\text{eqtzero}} : \quad & eq^s_t(0_\Gamma , 0_\Gamma,
                ?_A, ?_A : \; ?_\Gamma)) \leftarrow  
                    && \\
                \kappa_{\text{eqtsucc}} : \quad & eq^s_t(\sigma ?_{\iota_\Gamma} ,
                \sigma ?_{\iota'_\Gamma}, ?_A, ?_B : \; ?_\Gamma) \leftarrow  
                    eq^s_t(?_{\iota_\Gamma}, ?_{\iota'_\Gamma}, ?_A, ?_\Gamma)
                    && \\
                \kappa_{\text{eqtrefl}} : \quad & eq^s_t(?_a , ?_a, ?_A, ?_\Gamma) \leftarrow  
                    && \\
                \kappa_{\text{eqtelim}} : \quad & eq^s_t(?_{M_1} ?_{M_2} ,
                ?_{N_1} ?_{N_2}, ?_B, ?_\Gamma) \leftarrow  
                    eq^s_t(?_{M_1}, ?_{N_1}, \Pi ?_A . ?_B, ?_{\Gamma}) \land
                    eq^a_t(?_{M_2}, ?_{N_2}, ?_B, ?_{\Gamma})
                    && \\
                \kappa_{\text{eqtwhrl}} : \quad & eq^a_t(?_M , ?_N, ?_A,
                ?_{\Gamma}) \leftarrow  
                    whr(?_M, ?_{M'}), eq_t(?_{M'}, ?_{N}, ?_A, ?_\Gamma)
                    && \\
            \end{aligned}
        \end{gather*}
        \begin{gather*}
            \begin{aligned}[clr]
                \kappa_{\text{eqtwhrr}} : \quad & eq^a_t(?_M , ?_N, ?_A, ?_\Gamma) \leftarrow  
                    whr(?_N, ?_{N'}), eq^a_t(?_M, ?_{N'}, ?_A, ?_\Gamma)
                    && \\
                \kappa_{\text{eqtstr}} : \quad & eq^a_t(?_M , ?_N, ?_A, ?_\Gamma) \leftarrow  
                    eq^s_t(?_M, ?_N, ?_A, ?_\Gamma)
                    && \\
                \kappa_{\text{eqtexp}} : \quad & eq^a_t(?_M , ?_N, \Pi ?_A .
                ?_B, ?_\Gamma) \leftarrow  
                    \shifteq{?_M}{?_{M'}},
                    \shift{?_N}{?_{N'}} \land
                    eq^a_t(?_{M'} 0, ?_{N'} 0, ?_B, ?_A : \; ?_\Gamma)
                    && \\
                \kappa_{\text{eqsimpl}} : \quad & eq^a_t(?_M , ?_{M'}, ?_A ?_N
                , ?_\Gamma) \leftarrow  
                    eq^a_t(?_M, ?_{M'}, ?_A, ?_G)
                    && \\
                \kappa_{\text{whrs}} : \quad & whr((\lambda ?_A . ?_M) ?_N ,
                ?_{M'}) \leftarrow  
                    ?_M [?_N / 0] \equiv ?_{M'}
                    && \\
                \kappa_{\text{whrh}} : \quad & whr( ?_M ?_N , ?_{M'} ?_N) \leftarrow  
                    whr(?_M , ?_{M'})
                    && \\
            \end{aligned}
        \end{gather*}

        \noindent
    Finally, there are clauses that represent shifting and substitution on terms
    and types:
        \begin{gather*}
            \begin{aligned}[clr]
                \kappa_{\text{shiftTtintro}} : \quad & 
                    \shifteqI{ (\Pi ?_A . ?_M)}{\iota}{(\Pi ?_{A'} . ?_{M'})}
                    \leftarrow  
                    \shifteqI{?_A}{\iota}{?_{A'}}
                    {} \land {}
                    \shifteqI{?_M}{\sigma \iota}{?_{M'}}
                    && \\
                \kappa_{\text{shiftTtintro}} : \quad & 
                    \shifteqI{ (\lambda ?_A . ?_M)}{\iota}{(\lambda ?_{A'} . ?_{M'})}
                    \leftarrow  
                    \shifteqI{?_A}{\iota}{?_{A'}}
                    {} \land {}
                    \shifteqI{?_M}{\sigma \iota}{?_{M'}}
                    && \\
                \kappa_{\text{shifttelim}} : \quad & 
                    \shifteqI{ (?_M ?_N)}{\iota}{(?_{M'} ?_{N'})} 
                    {} \leftarrow {}
                    \shifteqI{?_M}{\iota}{?_{M'}}
                    {} \land {}
                    \shifteqI{?_N}{\iota}{?_{N'}}
                    && \\
                \kappa_{\text{shifttgt}} : \quad & 
                    \shifteqI{\iota}{0}{\sigma \iota}
                    \leftarrow  
                    && \\
                \kappa_{\text{shifttpred}} : \quad &
                    \shifteqI{0}{\sigma \iota}{0}
                    \leftarrow  
                    && \\
                \kappa_{\text{shifttstep}} : \quad & 
                    \shifteqI{\sigma \iota}{\sigma \iota'}{\sigma \iota''}
                    {} \leftarrow  {}
                    \shifteqI{\iota}{\iota'}{\iota''}
                    && \\
                \kappa_{\text{substTtintro}} : \quad & 
                    \substeqI{ (\Pi ?_A . ?_M)}{?_N}{\iota}{(\Pi ?_{A'} . ?_{M'})}
                    \leftarrow  
                    (\substeqI{?_A}{?_N}{\iota}{?_{A'}})
                    \land
                    (\shifteqI{?_N}{0}{?_N'})
                    \land
                    \substeqI{?_M}{?_N'}{\sigma \iota}{?_{M'}}
                    && \\
                \kappa_{\text{substintro}} : \quad & 
                    \substeqI{ (\lambda ?_A . ?_M)}{N}{\iota}{(\lambda ?_{A'} . ?_{M'})}
                    \leftarrow  
                    (\substeqI{?_A}{\iota}{?_{A'}})
                    \land
                    (\shifteqI{?_N}{0}{?_N'})
                    \land
                    \substeqI{?_M}{?_N'}{\sigma \iota}{?_{M'}}
                    && \\
                \kappa_{\text{substtelim}} : \quad & 
                    \substeqI{ (?_{M_1} ?_{M_2})}{?_N}{\iota}{(?_{M_1'} ?_{M_2'})} 
                    {}\leftarrow {}
                    \substeqI{?_{M_1}}{?_N}{\iota}{?_{M_1'}}
                    {} \land {}
                    \substeqI{?_{M_2}}{?_N}{\iota}{?_{M_2'}}
                    && \\
                \kappa_{\text{substz}} : \quad & 
                    \substeqI{0}{?_N}{0}{?_N}
                    \leftarrow  
                    && \\
                \kappa_{\text{substs}} : \quad &
                    \substeqI{0}{?_N}{\sigma \iota}{0}
                    \leftarrow  
                    && \\
                \kappa_{\text{substgt}} : \quad & 
                    \substeqI{\sigma \iota}{?_N}{0}{\sigma \iota}
                    \leftarrow  
                    && \\
                \kappa_{\text{substpred}} : \quad & 
                    \substeqI{\sigma \iota}{?_N}{\sigma \iota'}{\sigma \iota''}
                    \leftarrow  
                    \substeqI{\iota}{?_N}{\iota'}{\iota''}
                    && \\
            \end{aligned}
        \end{gather*}

\end{definition}

\noindent
The clauses in Definition \ref{def:Pempty} correspond to judgements in
Figures \ref{fig:wfnltypes}--\ref{fig:wfsigcon}.
They are direct
translations of the inference rules of nameless LF in these figures.
The judgement $\sig \vdash_{\text{Prog}} P$ extends $P_{e}$ with a clause
for each type and term constant in $\sig$
and initialises shifting and substitution with term and type-level
constants as constant under the operation.

\begin{definition}[Refinement Program]
    The judgement $\sig \vdash_{\text{Prog}} P$ is given by the inference rules of Figure
    \ref{fig:refsig}.
\end{definition}

%
%
%
%
%
%

\begin{figure}

    \framebox{$
        \sig   \vdash_{\text{Prog}}  P 
    $}\raggedright

    \vskip1em

    \noindent
    \parbox{0.15\textwidth}
    {
        \centering
        \AxiomC{$
        $}
        \UnaryInfC{$
            \cdot \vdash_{\text{Prog}} P_{\text{e}}
        $}
        \DisplayProof
    }
    \parbox{0.84\textwidth}
    {
        \centering
        \AxiomC{$
            \sig \vdash_{\text{Prog}} P
        $}
        \UnaryInfC{$
            \sig , c : A \vdash_{\text{Prog}} P 
            , \quad \kappa_c : \term(c, A , {}?_\Gamma ) \leftarrow 
            , \quad \kappa_{\text{shift}_{c}} : (\shifteqI{c}{0}{c}) \leftarrow ,
        $}
        \noLine
        \UnaryInfC{$
            \quad \quad 
            \kappa_{\text{subst}_c} : \substeqI{c}{ ?_M }{ 0 ]}{c} \leftarrow ,
            \quad \kappa_{eq^s_c} : eq^s(c, c, A, ?_\Gamma)
            \leftarrow
        $}
        \DisplayProof
    }
    \vskip1em
    \noindent
    \parbox{\textwidth}
    {
        \centering
        \AxiomC{$
            \sig \vdash_{\text{Prog}} P
        $}
        \UnaryInfC{$
            \sig , a : L \vdash_{\text{Prog}} P , \quad
            \kappa_{\text{shift}_{\alpha}} : ( \shifteqI{\alpha}{0}{\alpha} ) \leftarrow,
            \kappa_{\text{subst}_{\alpha}}  : \substeqI{\alpha}{ ?_M }{
                0}{\alpha} \leftarrow,
            \quad \kappa_{eq_T} : eq_T(\alpha, \alpha, L, ?_\Gamma) \leftarrow,
        $}
        \noLine
        \UnaryInfC{$
            \kappa_{eq^a{}_{\alpha}} : eq^a(?_N, ?_M, \alpha, ?_\Gamma) \leftarrow
                eq^s(?_M, ?_N, \alpha, ?_\Gamma)
        $}
        \DisplayProof
    }

    \caption{Refinement of signatures, with operations}
    \label{fig:refsig}
\end{figure}

\noindent
The Figure \ref{fig:refsig} gives definition of signature refinement.
The refinement judgement of a signature into a program concludes our
transformation of refinement problem into a goal and a program. 

\begin{theorem}[Decidability of Goal Construction]
  \label{thm:dec}
    Let $M$ be a refinement problem in a well-formed signature $\sig$ and a
    well-formed context
    $\Gamma$ such that a solution $(\rho, R)$ exists. Then inference rules
    in Figures \ref{fig:refterms} and \ref{fig:reftypes} construct the goal $G$
    and the extended type $A$ such that $\sig ; \Gamma ; M \vdash (G \mid A)$.
\end{theorem}
\begin{proof}
    By induction on the derivation of the well-formedness judgement of $(\rho, R)(M)$.
\end{proof}

The next example illustrates the construction of a refinement goal.

\begin{example}[From an Extended Nameless Term to a Goal]\label{ex:dec}
  
Let us take the refinement problem $M'
= (\elimmaybeA~\ttb~0)
(\lambda ?_A . ?_b)
$ and the implicit context and signature from Example~\ref{ex:rp}.
By Theorem \ref{thm:dec} we can generate $G$ such that the judgement
$\sig; \Gamma_1; M' \vdash ( G \mid {}?_{B_7})$ holds:
\begin{gather*}
    \footnotesize
    G = \top
    \land
    \top
    \land
    eq_T(
        \Pi \bool . (
            \Pi (\maybeA 0_T) . (
                \Pi (\Pi (2_T \eqb \ffb) . \A) . (
                    \Pi (\Pi (3_T \eqb \ttb) . (\Pi \A . \A))
                    . \A
                )
            )
        ),\\
        \Pi \bool . ?_{B_1},
        \Pi \typeK . ?_{L_1}, \Gamma_1
    )
    \land
        ?_{B_1}[\ttb/0_T] \equiv \; ?_{B_2}
    \land
    \top
    \land\\
    eq_T(?_{B_2}, \Pi (\maybeA \ttb) . ?_{B_3}, \Pi \typeK . ?_{L_2}, \Gamma_1)
    \land
    ?_{B_3}[0_\Gamma/0_T] \equiv \; ?_{B_4}
    \land\\
    \type(?_A, ?_{L_3}, \Gamma_1)
    \land
    ?_b : \term(?_b, ?_{A_1}, {?_A : \Gamma_1})
    \land
    eq_K(?_{L_3}, \typeK, \Gamma_1 ) 
    \land
    ?_{A_1} [ 0_T / 0_\Gamma ] \equiv \; ?_{B_5}
    \land\\
    eq_T(?_{B_4},\Pi (\Pi ?_{A} . ?_{B_5}) . ?_{B_6}, \Pi \type . ?_{L_5}, \Gamma) \land
    ?_{B_6} [ (\lambda ?_A . ?_b) / 0_T] \equiv \; ?_{B_7}
\end{gather*}
That is, the type of $M'$ will be computed as a substitution for logic variable
$?_{B_7}$ and resolving the goal computes
assignments to $?_A$ and $?_b$ as well.
\end{example}

\begin{proposition}[Decidability of Program Construction]
  \label{thm:decsig}
  Let $\sig$ be a signature.
  Then inference rules
    in Figure \ref{fig:refsig} construct the program $P$
    such that $\sig \vdash_{Prog} P$.
\end{proposition}

We develop our running example further to illustrate the proposition:

\begin{example}[From a Signature to a Program]\label{ex:decsig}
  The signature $\sig$ contains 
$\elimeqb$ hence the
generated program contains the clause: 
\begin{equation*}
    \kappa_{\elimeqb} :
    \term(
        \elimeqb,
        \Pi \ttb \eqb \ffb . \A,
        \typeK,
        ?_\Gamma)
    \leftarrow
\end{equation*}
The following clauses come from the program $P_e$ and represent inference rules of the internal language:
\begin{equation*}
    \begin{aligned}
        \kappa_{0}  & : \term(0, ?_A, ?_{A'} : \; ?_\Gamma) && \leftarrow  
            \; \shifteq{?_{A'}}{?_{A}}
            \\
        \kappa_{\elim}  & : \term(?_a ?_b, ?_B, ?_\Gamma) && \leftarrow 
        \term(?_a, \Pi ?_A . ?_{B'}, ?_\Gamma) \land
        \term(?_b, ?_A , ?_\Gamma) \land
        \\
        & && \quad \quad
        eq_T(?_B , ?_{B'}, \typeK, ?_\Gamma) \land
        \substeq{?_{B'}}{?_b}{?_B}
   \end{aligned}
\end{equation*}
\end{example}

\noindent
Example~\ref{ex:dec} shows unresolved meta-variables in the goal, and 
Example~\ref{ex:decsig} gives a program against which to resolve the goal.
Now the proof-relevant resolution comes into play.

\section{Proof-Relevant Resolution and Soundness}
\label{sec:sound}

We utilise a variant of proof-relevant resolution \cite{FuK17} as the inference
engine for solving refinement problems. A substitution of logic variables as
well as substitution composition is defined in the usual way. We use $\theta$,
$\theta'$ to denote substitutions and $\theta \circ \theta'$ to denote
composition of substitutions $\theta$ and $\theta'$. We use $[\delta_1 /
\gamma_1, \dots, \delta_n / \gamma_n]$ to denote an assignment that assigns, in
order, proof terms $\delta_1$ to $\delta_n$ to proof-term variables
$\gamma_1$ to $\gamma_n$. The resolution judgement
$P \vdash^\theta_{[ \delta_1 / ?_{\gamma_1},
\dots, \delta_n / ?_{\gamma_n} ]} G$
states that a goal $G$ is
resolved by a program $P$ while computing an answer substitution $\theta$ and an
assignment of proof terms. 
The
judgement makes use of an auxiliary judgement for resolution of atomic goals
with a proof term $\delta$, denoted
$P \vdash^\theta_{[ \delta_1 / ?_{\gamma_1}, \dots, \delta_n / ?_{\gamma_n} ]} \delta : At$.
\begin{definition}[Proof-Relevant Resolution]\label{def:prr}
    Let $P$ be a program and $G$ a goal, $At$, $At_1, \dots, At_n$ be atomic
    formulae, $\delta, \delta_1, \dots, \delta_n$ be proof terms
    and $\gamma_1, \dots, \gamma_n$ proof-term variables. The
    judgements $P \vdash^\theta_{[ \delta_1 / \gamma_1, \dots, \delta_n /
    \gamma_n ]}  G$ and
    $P \vdash^\theta \delta : At$ 
    are defined by the inference rules\\

    \parbox{4.8in}
    {
        \centering
        \AxiomC{$
            P \vdash^\theta \delta_1 : At_1
        $}
        \AxiomC{$
            \dots
        $}
        \AxiomC{$
            P \vdash^\theta \delta_n : At_n
        $}
        \TrinaryInfC{$
            P \vdash^{\theta}_{
                [ \delta_1 / {\gamma_1} , \dots, \delta_n / {\gamma_n} ]
        }
            {\gamma_1} : At_1
            \land
            \dots
            \land
            {\gamma_n} : At_n \quad
        $}
        \DisplayProof
    }

    \noindent
    and

    \parbox{4.8in}
    {
        \centering
        \AxiomC{$
            P \vdash^{\theta'} \delta_1 : \theta At_1
        $}
        \AxiomC{$
            \dots
        $}
        \AxiomC{$
            P \vdash^{\theta'} \delta_n : \theta At_n
        $}
        \LeftLabel{$\kappa : At' \leftarrow At_1 \land \dots \land At_n \in P$}
        \TrinaryInfC{$
            P \vdash^{\theta \circ \theta'}_{
        }
            \kappa \delta_1 \dots \delta_n : At
        $}
        \DisplayProof
    }

    \noindent
    assuming that $\theta At' = \theta At$.
\end{definition}

\noindent
We do not discuss a particular resolution strategy here, but instead refer the
reader to the work of \citeN{FuK17}.

Assume that $G$ and $P$ are a goal and a program that originate from a
refinement problem $M$ in signature $\sig$.
An answer substitution for $G$
computed by $P$ provides a solution to the type-level metavariables
in $M$. Similarly the computed assignment of
proof terms to proof variables provides a solution to the term-level metavariables in
$M$.

We continue with our running example, building upon Examples~\ref{ex:rp}--\ref{ex:decsig}.

\begin{example}[Proof-relevant Resolution Trace]\label{ex:decsig2}
  The resolution trace of our example is rather long, and we show only a fragment.
  Suppose that, in several resolution steps denoted by $\leadsto^*$, the goal $G$ given in Example~\ref{ex:dec} resolves as follows:
\begin{equation*}
G \leadsto^* \quad
    ?_b : \term(?_b, \A, \ttb\eqb\ffb : \Gamma_1)
  \end{equation*}

 \noindent  The computed substitution assigns $
    (\Pi (\Pi (\ttb \eqb \ttb) . (\Pi \A . \A)) . \A$ to the logic variable $?_{B_7}$.
  We now show the trace for
  the remaining goal $ ?_b : \term(?_b, \A, \ttb\eqb\ffb : \Gamma_1)$.
Given the clauses of Example~\ref{ex:decsig}, a resolution trace that computes a proof term for $?_b$ can be given as follows: 
\begin{equation*}
    \footnotesize
    \begin{split}
        \termSPLIT (?_a, \A, \ttb \eqb \ffb : \Gamma_1)
        \leadsto_{\kappa_{\elim}}\\
        & \term(?_{a_1} ?_{a_2}, \Pi ?_{A} . ?_{B'}, \ttb\eqb\ffb : \Gamma_1)
        \land
        \term(?_{a_{2}'}, ?_{A} , \Gamma_1)) \\
        & 
        \quad \quad
        \quad \quad
        \quad \quad
        \quad \quad
        \quad \quad
        \quad \quad
        \land
        eq_T(?_{B_4}, ?_{B'}, \typeK, \ttb\eqb\ffb : \Gamma_1)
        \leadsto
        _{\kappa_{\elimeqb}}
        ^{[?_{A} \mapsto \ttb \eqb \ffb, ?_{B'} \mapsto \A]}\\
        & \term(?_{a_2'}, \ttb \eqb \ffb, (\ttb \eqb \ffb : \Gamma_1)) \land
        eq_T( \ttb \eqb \ffb, \ttb \eqb \ffb, \typeK, \ttb \eqb \ffb : \Gamma_1))
        \leadsto_{\kappa_{0}} \\
        & eq_T( \ttb \eqb \ffb, \ttb \eqb \ffb, \typeK, \ttb \eqb \ffb : \Gamma_1))
        \leadsto^{*} \bot
    \end{split}
\end{equation*}
Above, we omit writing full proof terms, but record the witnesses as indices of $\leadsto$.
%
The assignment to the logic variable $?_A$ is $\A$ and
the subterm of the computed proof term that is bound to $?_b$ is
$
    \kappa_{\elim} \kappa_{\elimeqb} \kappa_{0} \delta_{eq_T}
$
where the subterm $\delta_{eq_T}$ is a witness of the appropriate type equality.
\end{example}

 Since we have used the types and terms of nameless LF to define our atomic formulae, the
computed substitution can be used directly. The interpretation of the computed
assignment of proof terms depends on assignment of atomic proof term symbols in
the program $P_e$. We define a mapping that gives the intended
interpretation:
\begin{definition}[Interpretation of Proof Terms]
    \label{def:in}
    We define interpretation of proof terms $\transl{-} : PT \to T$ as follows:

    \parbox{1.9in}
    {
    \begin{gather*}
        \begin{aligned}[c]
            \transl{ \kappa_\sigma \delta \delta_\iota \delta'} &= \sigma
                \transl{\delta_\iota} && \\
            \transl{ \kappa_{proj} \delta_\iota } &= \transl{\delta_\iota} && \\
            \transl{ \kappa_{\text{T-elim}} \delta_M \delta_N \delta \delta'}
                &= \transl{\delta_M} \transl{\delta_N} && \\
        \end{aligned}
    \end{gather*}
    }
    \parbox{2.3in}
    {
    \begin{gather*}
        \begin{aligned}[c]
            \transl{ \kappa_{\text{T-intro}} \delta_A \delta \delta_B }
                &= \Pi \transl{\delta_A} .
                \transl{\delta_B} && \\
            \transl{ \kappa_{\text{t-elim}} \delta_A \delta_M \delta \delta'}
                &= \transl{\delta_A} \transl{\delta_M} && \\
            \transl{ \kappa_{\text{t-intro}} \delta_A \delta \delta_M \delta' }
                &= \lambda \transl{\delta_A} .
                \overrightarrow{\transl{\delta_M}} && \\
        \end{aligned}
    \end{gather*}
    }
    \parbox{0.8in}
    {
    \begin{gather*}
        \begin{aligned}[c]
            \transl{ \kappa_{0} } &= 0 && \\
            \transl{ \kappa_{c}  }
                &= c \\
            \transl{ \kappa_{a}  }
                &= a .
        \end{aligned}
    \end{gather*}
    }
\end{definition}
%
We extend $\transl{-}$ to assignments of proof terms by composition and 
we use $\transl{R}$ to denote the composition ($\transl{-} \circ R$).
\begin{example}
  \label{ex:interp}
In Example~\ref{ex:decsig2}, the computed proof term bound to $?_b$ will be interpreted as follows:
\begin{equation*}
    \transl{
        \kappa_{\elim} \kappa_{\elimeqb} \kappa_{0} \delta_{eq_T}
    } = \elimeqb 0
\end{equation*}
Hence, the original problem is refined to $\elimmaybeA~\ttb~0~(\lambda \A .
\elimeqb 0)$ while the computed type is
    $((\ttb \eqb \ttb ) \to \A \to \A ) \to \A$
.
\end{example}
Finally, the above interpretation allows us to state the soundness of our system:

\begin{theorem}[Soundness of Proof-Relevant Resolution for Generated Goals and Programs]\label{th:sound}
    Let $M$ be a term in the extended syntax with signature $\sig$. Let
    $P$ and $G_M$ be a program and a goal such that $\sig, \cdot
    \vdash (G_M | A)$ and $\sig \vdash_{\text{Prog}} P$ respectively. Let
    $\rho$, $R$ be a substitution and a proof term assignment computed by proof-relevant resolution such that
    $P \vdash^\rho_R G_M$.
    Then if there is a solution for a well-formed term, then there are 
    solutions $(\rho', R')$ and $(\rho'', R'')$ 
    such that $(\rho',R')M$ is a well-formed term and
    $$(\rho'', R'')((\rho, \transl{R}) M) = (\rho', R') M$$
\end{theorem}
\vskip-1em
\begin{proof}
    Generalise the statement of the theorem for an arbitrary well-formed context
    $\Gamma$. By simultaneous induction on derivation of the well-formedness
    judgement of $(\rho', R')M$ and derivation of $P \vdash^\rho_R G$. The
    theorem follows from the generalisation.
\end{proof}

Theorem~\ref{th:sound} guarantees that 
the refinement computed in Examples~\ref{ex:dec}--\ref{ex:interp} is 
well-typed in the internal language. That is, there is a derivation of the following judgement:
\begin{equation*}
    \sig ; \cdot, \maybeA \ttb \vdash 
        \elimmaybeA~\ttb~0~(\lambda \ttb\eqb\ffb .  \elimeqb 0) 
        : (\ttb \eqb \ttb ) \to \A \to \A ) \to \A
\end{equation*}
We omit the actual derivation of the judgement. However, note that it can be
easily reconstructed in a similar way 
as the intended
interpretation of proof terms is computed in Definition~\ref{def:in}. For example, in case of our running
example, the subterm $\delta_{eqT}$ of the proof term gives derivation of the
definitional equality that is necessary to verify application of $\elimeqb$ to
index $0$.

\subsection*{Implementation}

We have formalised the results in this paper using the Ott tool and the Coq theorem prover. The
source code can be found 
online\footnote{\url{https://github.com/frantisekfarka/slepice}}.
Most importantly, we formalise definitions of nameless LF.
The exported Coq definitions are then used in formalisation of decidability
of the refinement calculus. 
An implementation of translation from the extended language to logic programs and
goals is obtained from the formal
proof via code extraction into OCaml. A parser is extracted from the
formalisation of the grammar as well. The translation outputs logic programs and
goals suitable for an external resolution engine.

\section{Related and Future Work}
\label{sec:relwork}

Ideas underlying our work originate in the work of \citeN{StuckeyS02} on HM(X) type
inference as (constraint) logic programming. There are two key differences.
First, in our work we consider dependent types. Other approaches, such as that
of Sulzmann and Stuckey
do not give a motivation for the shape of
generated logic goals and programs.
We make explicit that
atomic formulae represent judgements of the type theory and that the program
originates on one hand from inference rules of the type theory and on the other
from a signature of a term. We believe that a clear identification of this
interpretation of generated goals and programs makes it feasible to adjust the
refinement calculus for different type theories.


Currently implemented systems (\cf \citeN{Pientka13}) make use of a bidirectional
approach to type checking. That is, there are separate type checking and
type synthesis phases.
The key difference between these
systems and our own work is that we do not explicitly discuss
bidirectionality. Combining this with a clear identification of atomic
formulae with judgements, and Horn clauses with inference rules,
in our opinion, makes the
presentation significantly more accessible. However,
bidirectionality in our system is still implicitly present, albeit postponed to
the resolution phase. As future work, we intend to analyse structural
resolution \cite{FuK17} for the generated goals. We intend to show that the matching
steps in the resolution correspond to type checking in the bidirectional approach
whereas resolution steps by unification correspond to type synthesis. 

In future work,
we would like to consider additional constructs in the
surface language. One example of such a construct would be a type-class mechanism, as found
in \eg Haskell.
\citeN{FuKSP16} have previously demonstrated that type class resolution in
Haskell can be addressed by proof-relevant resolution in Horn clause logic with
explicit proof terms. This result suggests a
natural extension of our work by adding a new atomic formula that represents
type-class judgement and by adjusting refinement calculus with inference rules
for the translation of type-class judgements. Recently, \citeN{BottuKSOW17} argued for
quantification on type-class constraints.  Although such quantification
escapes the Horn clause fragment as discussed in the work of Fu \etal, it can be
addressed in the logic of hereditary Harrop formulae. Our approach scales well by
replacing the target logic by the logic of first-order hereditary Harrop
formulae. Finally, we believe that our suggested approach to refinement can also
serve as a viable method for proving the soundness of the surface language w.r.t. a
semantic model. We have already presented some initial results \cite{FarkaKH16} that show
that proof-relevant type class
resolution is sound w.r.t. to both 
inductive and coinductive interpretation of type class instances.

\section{Conclusion}
\label{sec:conclusion}

Functional programming languages are a convenient tool for developing
software.  Dependent types in particular allow various semantic
properties to be encoded in types.  However, as evidenced by languages
such as Agda, Coq or Idris, software development in a functional
programming language with dependent types is a complex task. The usability
of such languages critically depends on the amount of assistance that an
automated tool provides to a programmer: at a minimum, some 
type refinement is necessary. Current descriptions are implementation
dependent and hard to understand. This complicates the reuse of existing
approaches in the development of tools for new languages.  Ultimately, it
is problematic for a programmer as they need to understand the
elaboration process to some extent.  We present a description of
refinement in LF that is significantly simpler than the existing ones. We
show a translation of an incomplete term with metavariables to a goal
and a program in Horn clause logic by a syntactic traversal of the
term. The inference  is then performed by proof-relevant resolution that is an
extension of the standard first-order resolution with proof term construction. Moreover, the
generated goal and program have a straightforward interpretation as
judgements of type theory and inference rules and
hence can be easier to understand and work with.

Our discussion of related work shows how our approach links to the state of the art in
refinement in functional languages. We have suggested two different areas for future
work. Firstly, a more detailed analysis of resolution on generated goals and
programs may recover bidirectionality. Secondly, we have discussed a possible extension
to the surface language, higher order features and type classes, and argued that
our approach scales well with extensions of the surface language. However, 
extensions of the surface language may require a stronger logic as
a target logic of the refinement calculus.


\paragraph{Acknowledgements}
{ \footnotesize
This work has been supported by the EPSRC grant
``Coalgebraic Logic Programming for Type Inference'' EP/K031864/1-2, EU
Horizon 2020 grant ``RePhrase: Refactoring Parallel Heterogeneous Resource-Aware
    Applications - a Software Engineering Approach'' (ICT-644235), and by
COST Action IC1202 (TACLe), supported by COST (European Cooperation in
Science and Technology).
%
We would like to thank Andreas Abel,
Greg Michaelson, and anonymous referees
for helpful comments on
earlier versions of the paper.
}

\bibliographystyle{acmtrans}
\bibliography{references}








\end{document}